%% file: 0main-arxiv.tex
\documentclass{article}
\usepackage[margin=1in]{geometry}
\usepackage[hidelinks,pdfencoding=auto]{hyperref}

\usepackage{times}

\usepackage{soul}
\usepackage{url}
\usepackage[utf8]{inputenc}
\usepackage[small]{caption}
\usepackage{graphicx}
\usepackage{amsmath}
\usepackage{booktabs}
\usepackage{enumerate}
\usepackage{authblk}
\usepackage{comment}

\usepackage{tikz}
\usetikzlibrary{arrows, decorations.pathmorphing}
\usetikzlibrary{arrows.meta, positioning, decorations.markings}
\usepackage[ruled]{algorithm2e}
\usepackage{amssymb}
\usepackage{amsthm} 
\usepackage{multicol}

\tikzset{
	photon/.style={decorate, decoration={snake}, draw=red}}

\newcommand{\MIS}{{\sc MIS}}

\newcommand{\NP}{{\sf NP}}

\newcommand{\gmis}{{\sc generalized maximum independent set}}

\newcommand{\CMM}{{\sc CMM}}
\newcommand{\E}{\mathop{\mathbb{E}}}

\newcommand{\Ca}{C}
\newcommand{\Cp}{C}

\newtheorem{theorem}{Theorem}
\newtheorem*{theorem*}{Theorem}

\newtheorem{lemma}{Lemma}

\newtheorem{proposition}{Proposition}

\theoremstyle{definition}
\newtheorem{model}{Model}
\newtheorem{fact}{Fact}

\theoremstyle{remark}
\newtheorem{remark}{Remark}

\newtheoremstyle{example}{}{}{}{}{\bfseries}{\smallskip}{\newline}{}
\theoremstyle{example}

\newif\ifFull
\Fulltrue

\newenvironment{onlyfullversion}{}{}
\newenvironment{onlyshortversion}{}{}

\newif\ifShort

\ifFull 
\Shortfalse
\includecomment{onlyfullversion}
\excludecomment{onlyshortversion}
\else
\Shorttrue
\excludecomment{onlyfullversion}
\includecomment{onlyshortversion}
\fi
\newcommand*\samethanks[1][\value{footnote}]{\footnotemark[#1]}

\title{Matchings with Group Fairness Constraints: Online and Offline Algorithms}

\author[1]{Govind S. Sankar}
\author[2]{Anand Louis\footnote{These three authors contributed equally.}}
\author[1]{Meghana Nasre\samethanks}
\author[3,4]{Prajakta Nimbhorkar\samethanks}
\affil[1]{Indian Institute of Technology Madras, Chennai}
\affil[2]{Indian Institute of Science, Bangalore}
\affil[3]{Chennai Mathematical Institute, Chennai}
\affil[4]{UMI ReLaX}
\affil[ ]{	govindbose@gmail.com,
	anandl@iisc.ac.in,
	meghana@cse.iitm.ac.in,
	prajakta@cmi.ac.in}
\date{}

\begin{document}

	\maketitle 
	
	\begin{abstract}
		We consider the problem of assigning items to platforms in the presence of group fairness constraints. In the input, each item belongs to certain categories, called {\em classes} in this paper. Each platform specifies the group fairness constraints through an upper bound on the number of items it can serve from each class. Additionally, each platform also has an upper bound on the total number of items it can serve. The goal is to assign items to platforms so as to maximize the number of items assigned while satisfying the upper bounds of each class. In some cases, there is a revenue associated with matching an item to a platform, then the goal is to maximize the revenue generated.
		
		This problem models several important real-world problems like ad-auctions, scheduling, resource allocations, school choice etc.We also show an interesting connection to computing a generalized maximum independent set on hypergraphs and ranking items under group fairness constraints.

		We show that if the classes are arbitrary, then the problem is \NP-hard and has a strong inapproximability. We consider the problem in both online and offline settings under natural restrictions on the classes. Under these restrictions, the problem continues to remain \NP-hard but admits approximation algorithms with small approximation factors. We also implement some of the algorithms. Our experiments show that the algorithms work well in practice both in terms of efficiency and the number of items that get assigned to some platform. 
	\end{abstract}

	\input{1-new-intro.tex}
	\input{3approx.tex}
	\input{4online.tex}
	\input{7experiments.tex}
	\input{5future.tex}
	
	
	\section*{Acknowledgements}
	We acknowledge some initial discussions with Ajay Saju Jacob.
	We are grateful to the anonymous reviewers for their comments. 
	AL was supported in part by SERB Award ECR/2017/003296 and a Pratiksha Trust Young Investigator Award. 
	MN and PN are supported in part by SERB Award CRG/2019/004757. 
	\bibliographystyle{amsalpha} 
	\bibliography{references}
	
	
\end{document}

%% file: 1-new-intro.tex
\newcommand{\Anote}[1]{{\color{red} [#1 --- Anand.]}}
\newcommand{\pnote}[1]{{\color{cyan} [#1 -- Prajakta.]}}
\newcommand{\mnote}[1]{{\color{blue} [#1 -- Meghana.]}}
\newcommand{\gnote}[1]{{\color{blue} [#1 -- Govind.] }}
\newcommand{\gedit}[1]{#1}

\newcommand{\Item}{item}
\newcommand{\Items}{items}
\newcommand{\Platform}{platform}
\newcommand{\Platforms}{platforms}
\newcommand{\cC}{\mathcal{C}}

\section{Introduction}

Graph matching is a fundamental problem in graph theory and theoretical computer science 
that has been studied extensively over the years. 
Computing the maximum matching in bipartite graphs, both in the online and the offline setting
is an important  building block in many applications in allocation problems such as
ad-auctions \cite{mehta_online_survey,mehta_online_adwords}, scheduling \cite{venkat_scheduling}, resource allocation \cite{halabian_resourceallocation},  school choice~\cite{abdulkadiroglu2003} etc. 
Since the notation used in these various problems differ, we use the general terms {\em \Items} and {\em \Platforms} to refer to the two parts of the bipartite graph. 
In practice, items may be classified based on different properties and hence may belong to certain groups or classes.
Modeling the allocation problems as a vanilla matching problem  seeks
to optimize the cost of the solution alone and may inadvertently be ``unfair'' to some
classes of \Items. 
Necessitated by the need to be fair and unbiased towards any group of \Items\ 
in the input, there has been a lot of recent work studying algorithms for
various problems augmented with fairness constraints,
such as \cite{vishnoi_fairness,KMM15,CHRG16,BCZSK16}).

In this paper, we enforce {\em group fairness} through constraints that place an upper bound on the number of items that can be
matched from a particular class to a platform. 
\gedit{We note that group fairness constraints usually involve both upper and
lower bounds. 
This is incompatible with the practical applications that we have in mind, namely 
ad-allocation and course allocation problems. 
For this reason, we focus only on upper bounds. }
We formally define the problem as follows. 

\paragraph{Classified maximum matching (\CMM)}
We have a set $A$ of items and a set $P$ of platforms, and these sets form 
the two parts of a bipartite graph. The presence of an edge $(a,p)$ indicates that item $a$ 
can be matched to platform $p$. Let $N(p)$ denote the neighborhood of $p$.
\textcolor{black}{
Each platform $p$  has a collection of classes $\cC_p \subseteq 2^{N(p)}$, i.e., each item
in $N(p)$ may belong to some of the classes in $\cC_p$.
Each class $C \in \cC_p$ has an associated quota $q_p^C$ denoting the maximum number
of items from $C$ that can be assigned to $p$.
In addition, each platform $p$ has a quota $q_p$, which is an upper bound on the total number of items
it can serve.}
Our goal is to compute an assignment of items to platforms so as to maximize the number 
of items assigned, while satisfying all the quotas.

Classes allow platforms to specify group fairness constraints -- for instance the classes can be seen as properties or categories
and the quotas impose the constraints that not too many items from one category are assigned to a platform.
These types of fairness constraints have been studied in many practical applications.
In ~\cite{abdulkadiroglu2003}, the authors address the school choice problem where fairness constraints are
imposed to achieve racial, ethnic, and gender balance. In the assignment of medical residents to hospitals in Japanese Residency Matching Program (JRMP),
regional quotas are introduced to ensure fairness amongst urban and rural hospitals \cite{kamada2012_strategyproof,kamada2015_distributional}. 
Huang~\cite{Huang10} motivates classifications from the perspective of enabling diversity in academic hiring. Apart from matchings, group fairness 
constraints have also been studied for many other problems like the knapsack problem \cite{anand_knapsack}, and clustering problems \cite{bera_fair_clustering}, to name a few. 

\gedit{
	Some recent pre-prints have discussed fairness in matching problems. 
	Soriano and Bonchi \cite{soriano2020fairbydesign} study a different notion of 
	individual fairness that they call maxmin-fairness. Their goal is to 
	output a solution such that the \emph{satisfaction} of one agent cannot be 
	improved without making another agent worse-off. 
	Ma and Xu \cite{ma2020grouplevel} measure fairness by the 
	ratio of expected number of agents matched from a particular group to the expected number of agents from that group and 
	their goal is to maximize the minimum of this ratio over all groups. 
	Basu \textit{et~al}. \cite{basu2020framework} also measure fairness based on metrics involving 
	the ratio of agents across groups and the utility they provide. 
	While qualitatively similar, our constraints can be seen as being 
	orthogonal to such notions of fairness.
}




In different applications, the fairness constraints can have a different structure. In \cite{kamada2012_strategyproof,kamada2015_distributional}, each hospital belongs to
exactly one region, hence fairness constraints partition the set of hospitals. On the other hand, in the school choice problem in \cite{abdulkadiroglu2003}, a student belongs to multiple constraints. The structure of the constraints crucially affects the computational complexity
of finding a fair allocation. This has been illustrated in the context of bipartite matchings where one or  both sides
of the bipartition express preferences over the other side. Huang~\cite{Huang10} and Fleiner and Kamiyama~\cite{FK16}, address the \CMM\ problem when both sides
of the bipartition have preferences over each other and the notion of optimality is {\em stability}, whereas 
\cite{nasre_classified_matchings} study the \CMM\ problem 
under the optimality criteria of {\em rank-maximality} and {\em popularity}. In all these cases, it has been shown that the \CMM\ 
problem can be efficiently solved if the constraints have a {\em laminar}\footnote{
A family $\mathcal{C}$ of subsets of a set $S$ is laminar if, for every pair of sets $X, Y \in \mathcal {C}$,
either $X\subseteq Y$ or $Y\subseteq X$ or $X\cap Y=\emptyset$.}
structure, and is \NP-hard in general \cite{nasre_classified_matchings}. 
Such a restriction has been considered before in the literature, such as in the hospital-resident problem \cite{kamada2012_strategyproof} or the college admissions problem \cite{biro2010_college_admissions,goto2016_regional}.
However, a finer relation between the structure 
of the class constraints and the computational efficiency has not received much attention in literature. In this paper, we address this issue by focusing on a quantification of non-laminarity in the classes and its effect on computational efficiency.
We strengthen the  hardness results in \cite{nasre_classified_matchings} and obtain new approximation algorithms for the \CMM\ problem in the offline setting.

Next, we turn our attention to the online version of the problem.  Online matching problems have numerous practical applications, such as in 
ad-allocations \cite{mehta_online_adwords}, resource allocations \cite{devanur_online_allocation}, etc. 
See \cite{mehta_online_survey} for a survey on online bipartite matchings. Fairness has also been studied in online settings such as online learning \cite{gillen2018_online_learning} or ride-hailing platforms \cite{suhr2019_ride_hailing}. 
`Fairness' in another form has been considered previously in the online literature. For example, the `frequency caps' mentioned in \cite{mehta_online_iid} places a restriction on the total number of ads that are shown to the same user, or users from a particular demographic.
We study some natural online versions of the \CMM\ problem. We first show that one of our approximation algorithms for the offline non-laminar case also works as an online algorithm, regardless of the input model. For the setting where we restrict classes to be laminar, we show that existing algorithms for online bipartite matching carry over to our setting.

\subsection{Models}

We study \CMM~problem in various settings.
In practice, assigning an item $a$ to a platform $p$ may generate a revenue, which can be modelled as the weight of the edge $(a,p)$. In such a case the goal of the weighted \CMM~problem  is to compute an assignment of items to platforms so as to maximize the total weight of edges in the matching, while satisfying the quotas of all the classes.
We now formally define our models. 
\begin{model}[Many-to-one]\label{mod:mo}
This is the setting described earlier. In this setting, items can match to at most one platform.
\end{model}

\begin{model}[Many-to-many]\label{mod:mm}
This is a more general setting in which items may be matched to multiple platforms.
In addition to the classes of platforms, each item $i$ also has a collection of classes $\cC_i \subseteq 2^{N(i)}$, i.e., each platform
in $N(i)$ belongs to some of the classes in $\cC_i$, and the items also have quotas for their classes.
This model arises in scenarios like course allocation, where students may be allotted multiple courses subject to various restrictions. Courses may have restrictions over the number of students allotted to it from each department or from each batch. 
\end{model}

In the setting of online matchings, the platforms are available offline and the items arrive online. When an item $a\in A$ arrives online, its neighbours in $P$, and the classes that it participates in are revealed. It must be immediately decided if we match $a$ to some platform and any edges matched cannot be unmatched later. 
In the literature, the order in which the items arrive has been studied in various models. 
In the adversarial model, the items can arrive in an arbitrary order. 
We study a natural online arrival model for the \Items, called the Random input model. 
See \cite{mehta_online_survey} for a survey of other work on such models. 

\begin{model}[Random input \cite{mehta_online_random}]\label{mod:random}
In this setting, there is an underlying graph $G=(A\cup P,E)$. The vertices of $A$ arrive according to a permutation chosen uniformly at random.
\end{model}

\subsection{Our results}
\label{sec:results}

In most applications, an item typically belongs to small number of classes, hence we first study this setting. For example, \cite{abdulkadiroglu2003} discusses the Boston student assignment mechanism which divides students into two categories based on whether they already have a sibling in the school and whether they are within walking distance to the school. 
Similarly, in a faculty hiring scenario, the number of classes (which would correspond to specializations) is independent of the number of applicants. 
For the scenario when there are constant number of classes we show the following result.
\begin{theorem}[Informal version of Theorem~\ref{thm:constant_classes_formal}]
\label{thm:constant-classes}
The \CMM~ problem can be solved in polynomial time if there is a constant number of platforms, each with a constant number of classes. 
This leads to a $\frac{1}{2}$-approximation algorithm for an arbitrary number of platforms, each with a constant number of classes. 
\end{theorem}
Now we turn to a more general setting where the number of classes is arbitrary
and exploit the structural relation amongst the classes.
We know  from \cite{nasre_classified_matchings} that when
the classes of every platform form exactly one laminar family then the \CMM\ problem is solvable in polynomial time.
We prove the following theorem. 
{
\begin{theorem}\label{thm:approx-delta-lam}
	There is a polynomial-time algorithm achieving an approximation ratio of $\frac{1}{\Delta+1}$ 
	for the many-to-one setting (Model \ref{mod:mo}) where each item belongs to at most $\Delta$ laminar families of classes per platform. This generalizes to the \gedit{weighted} many-to-many setting (Model \ref{mod:mm}) where for each edge $(a,p)$, the classes of $a$ and $p$ that contain 
	$(a,p)$ can be partitioned into $\Delta+1$ laminar families.
\end{theorem}
}
The above result is applicable in scenarios like ad-allocation where the number of classes can be arbitrary, but any ad belongs to
a few of them.
Complementing this, we also obtain hardness results for computing the optimal \CMM. 
\begin{proposition}
\label{prop:hardness-indset1}
\begin{enumerate}[(i)]
\item \label{itm:general} \CMM\  cannot be 
approximated to a factor of $n^{\epsilon-1}$ for any $\epsilon>0$ unless ${\sf P}=\NP$, where $n=|A|$, even when there is 
a single platform and all edge weights are one.
\item \label{itm:bounded-degree} When there is a single platform, and additionally, each item appears in at most $\Delta$ classes, the problem
is {\sf NP}-hard to approximate within a factor $O\left(\frac{\log^2 \Delta}{\Delta}\right)$.
\end{enumerate}
\end{proposition}
The proof of Proposition \ref{prop:hardness-indset1} follows from a 
reduction from the {\sf maximum independent set} problem.

\paragraph{Online algorithms.} 
We first remark that our algorithm from Theorem \ref{thm:approx-delta-lam} works as an online algorithm in the \gedit{unweighted case}, even when the input is adversarially chosen.
\begin{theorem}\label{thm:approx-online-delta-lam}
	There is a polynomial-time online algorithm achieving, in any input model, a competitive ratio of $\frac{1}{\Delta+1}$ for the many-to-one setting (Model \ref{mod:mo}) where each item belongs to at most $\Delta$ laminar families of classes per platform. The algorithm extends to the many-to-many setting (Model \ref{mod:mm}) where the classes of $a$ and $p$ containing each edge $(a,p)$ can be partitioned into $\Delta+1$ laminar families.
\end{theorem}

Having achieved a competitive ratio that is close to the lower bound (from Proposition~\ref{prop:hardness-indset1}~(\ref{itm:bounded-degree})), 
we consider the case where classes are restricted to be laminar. We consider the random order input model (Model~\ref{mod:random}) and show that a simple greedy algorithm from the literature also works for \CMM~ and that it achieves the same competitive ratio. We use the technique of randomized dual fitting which has been used to analyse competitive ratios in works such as \cite{devanur_ranking,huang_deadlines}. 
\begin{theorem}\label{thm:approx-online-random}
There is a polynomial-time algorithm achieving a competitive ratio of $1-\frac{1}{e}$ for \CMM\ with laminar classes in the random input model.
\end{theorem}

\begin{onlyshortversion}
	\subsection{Implications for other problems}\label{sec:conseq}
\end{onlyshortversion}
\begin{onlyfullversion}
	\section{Implications for other problems}\label{sec:conseq}
\end{onlyfullversion}
Although the \CMM\ problem is modelled as a matching of items to platforms, we show that the
classes capture  problems which are well studied and are of independent interest. 

\paragraph{Maximum Independent set on hypergraphs}
Given a hypergraph $H = (V,E)$, and a function $f: E \to \mathbb{Z}^+$, compute the largest set of vertices $S$ such that 
for every  $e \in E$, $|S\cap e|\leq  f(e)$.
We note that when $f(e) = 1 \ $ for each edge $e$, this is the problem of computing the {\em strong}
maximum independent set and when $f(e) = |e|-1$, this is the {\em weak} maximum independent set problem.
These problems are well-studied for bounded-degree hypergraphs; \cite{halldorson_hypergraph_ind} describe algorithms achieving factors of {$\frac{1}{\Delta}$ and
$\frac{5}{\Delta +3}$} for the strong and weak cases respectively, where $\Delta$
denotes the maximum degree of a vertex in $H$. For the weak independent set, this was further improved to $O\left ( \frac{\log \Delta}{\Delta \log \log \Delta}\right )$ in \cite{halldorsson_sdp}.
However, to the best of our knowledge, there is no known approximation algorithm for the case when $f(e)$ is an arbitrary value -- we call this 
the \gmis~ on hypergraphs. We state our approximation result for independent sets on hypergraphs below, which is a consequence of Theorem~\ref{thm:approx-delta-lam}.
\begin{proposition}
\label{prop:gmis}
There is a polynomial time $\frac{1}{\Delta}$ approximation algorithm for the problem of computing a \gmis~on hypergraphs with maximum degree $\Delta$.
\end{proposition}
For the case when average degree of the vertices is $\Delta$, we get the following:
\begin{theorem}\label{thm:approx-avg-degree}
	There is a $\frac{r}{4\Delta}$ approximation algorithm for the generalized independent set where $r = \frac{OPT}{n}$ and
	$\Delta$ denotes the average degree of a vertex.
\end{theorem}
 
For the CMM problem, this implies an $\frac{OPT}{4\Delta n}$ approximation algorithm when we only have an upper bound on the average number of laminar families of classes an item belongs to, and there is only one platform.
\paragraph{Ranking and group fairness} In an apparently different model, Celis~et~al.~\cite{vishnoi_fairness} consider
ranking $n$ items from a universe of $m$ items, where $n \ll m$. Items are assigned properties, 
and upper quotas for the number of items from any property in the top $k$ ranks. 
When items have at most $\Delta$ properties each, they give a $\frac{1}{\Delta+2}$ 
approximation while allowing constraints to be violated by a factor of $2$. This problem can be reduced to the \CMM\ problem
and our algorithm from Theorem~\ref{thm:approx-delta-lam} achieves the same approximation factor without violating any class constraints.
\begin{onlyshortversion}
	We leave the reduction to the full version of this paper~\cite{full}.
\end{onlyshortversion}
\begin{onlyfullversion}
	
We now show a reduction from their problem to the \CMM\ problem.
The fair rankings problem is defined as follows: Given a set of $m$ items, our objective is to choose $n$ among them and rank them. Every item $i$ has a value $W_{ij}$ when ranked in the $j^{th}$ position. Every item also has some properties, and every property can be represented as a subset of the $m$ items that share that property. The objective is to rank them such that this the total value is maximized, subject to some fairness constraints on the properties. A fairness constraint corresponding to the property $p$ is set of values $U_1^p,U_2^p,\ldots, U_n^p$ such that $U_k^p$ is the maximum number of items with property $p$ that can be ranked in the top $k$ positions.

In the reduction, we create one platform $p$. For every item $i$ in the ranking instance we have $n$ items $(i,1),(i,2),\ldots, (i,n)$ in the \CMM\ instance. The item $(i,j)$ being matched to $p$ will be equivalent to ranking the $i^{th}$ item in the $j^{th}$ position. Then our classes are
\begin{itemize}
	\item One class for every item to ensure that an item is ranked at most once. That is, $\forall\ i\in \{1,2,\ldots, m\}$, we have the class
	\[
	C= \{(i,1),(i,2),(i,3)\ldots (i,n)\} \quad q(C) = 1
	\]
	\item One class for every rank to ensure that at most one item is ranked in one position. That is, $\forall\ j\in \{1,2,\ldots, n\}$, we have the class
	\[
	C= \{(1,j),(2,j),(3,j)\ldots (m,j)\} \quad q(C) = 1
	\]
	\item One laminar family of classes for every property $p$ with constraints $U_1^p,U_2^p,\ldots, U_n^p$. There are $n$ constraints, one for each position in the ranking. Let $p=\{i_1,i_2,\ldots, i_k\}$  be the items with property $p$. Then $\forall\ p,\forall j\in\{1,2,\ldots, n\}$ we have
	\begin{align*}
		C_j^p= \{(i_1,j),(i_2,j),(i_3,j)\ldots(i_k,j)\} \cup C_{j-1}^p \\
		q(C_j^p) = U_j^p
	\end{align*}
	where $C_0^p =\{\}$. It is easily seen that this is a laminar family.
\end{itemize}
Thus, for an item $i$ in the ranking instance that belongs to $\Delta$ properties, for any $j\in \{1,2,\ldots, n\}$, $(i,j)$ in the constructed \CMM\ instance belongs to $\Delta+2$ laminar families of classes. Using the algorithm from Theorem $\ref{thm:approx-delta-lam}$, we get a $\frac{1}{\Delta+2}$ approximation without any violation in quotas. However, it has to be noted that in \cite{vishnoi_fairness}, they insist on finding a ranking with $n$ elements whereas we may output a possibly smaller ranking. We also note that this reduction immediately gives a better hardness of approximation of $O\left(\frac{\log \Delta}{\Delta}\right)$ for the setting where every item lies in $\Delta$ laminar families of classes per platform.

\end{onlyfullversion}

\paragraph{Simultaneous matchings} Kutz~et~al.~\cite{mahajan_simultaneous_matchings} study the problem called simultaneous matchings
which is defined as follows:
given a bipartite graph $G=(X\cup D, E)$ and a collection $\mathcal{F}\subseteq 2^X$, 
find the largest solution $M \subseteq E$ such that $\forall\ C\in \mathcal{F}, M\cap (C \times D )$ is a matching. 
This problem can be reduced to the \CMM\ problem where every vertex $d$ in $D$ has constraints $\mathcal{F}$ (excluding vertices to which $d$ has no edge), and each class has quota $1$. The approximation factor in \cite{mahajan_simultaneous_matchings} is better but the constraints are significantly more restricted than ours.	

%% file: 3approx.tex
\section{Offline Approximation algorithms}
\label{sec:approx}

\begin{onlyshortversion}
	In light of the strong inapproximability result for the general \CMM\ problem (Proposition~\ref{prop:hardness-indset1}), we describe approximation algorithms for some special cases.  
\end{onlyshortversion}

\begin{onlyfullversion}
We start by showing the hardness of the \CMM\ problem. 
\begin{proof}[Proof of Proposition~\ref{prop:hardness-indset1}]
	Consider an arbitrary instance of the maximum independent set (\MIS) problem, which is an undirected graph $G=(V,E)$. From $G$, we will create an instance of the \CMM\ problem denoted as a graph $H=(A\cup\{p\},E')$, and a set of platform classes $\mathcal{C}$, where $A$ is the set of items and $p$ is a single platform.
		\begin{eqnarray}
			A &=&\{a_i\mid v_i\in V\}\nonumber \\
			\mathcal{C} &=& \{C_{ij}\mid (v_i,v_j)\in E\}\nonumber \\
			E' &=&\{(a_i,p)\mid a_i\in A\}\nonumber
		\end{eqnarray}
		Thus, there is an item $a_i\in A$ for each $v_i\in V$, and every item in $A$ has an edge to $p$. There is a class $C_{ij}$ with quota $1$ for every edge
		$(v_i,v_j) \in E$. We claim that the instance $H$ so constructed has a feasible matching of size $k$ if and only if $G$ has an independent set of size $k$.
		
		Since an independent set consists of at most one end-point of an edge $(v_i,v_j)$, the corresponding set of items respects quota of each class $C_{ij}$.
		Thus, given an independent set of size $k$, there is a feasible matching of size $k$ in the \CMM\ instance. 
		Similarly, given a feasible matching of size $k$ in the \CMM\ instance, the set of vertices corresponding to the matched items form an independent set in $G$. There cannot be two matched items $a_i,a_j$ such that $(v_i,v_j)\in E$ because of the quota of the class $C_{ij}$. Part (\ref{itm:general}) of Proposition~\ref{prop:hardness-indset1} follows from the hardness of approximation of \MIS~\cite{zuck_mis_hardness}.

		The \MIS\ problem is known to be \NP-hard even when the input instance consists of a graph where the maximum degree of any vertex is at most $\Delta$, where $\Delta\geq 3$ \cite{garey_johnson_npcomplete}. It is also known to be {\sf NP}-Hard to approximate below a factor of $O(\frac{\Delta}{\log^2\Delta})$ when $\Delta=O(1)$~\cite{bhangale_mis_hardness}. Part (\ref{itm:bounded-degree}) of Proposition~\ref{prop:hardness-indset1} follows from this, by the same reduction, since the degree of a vertex $v_i$ in the \MIS\ instance is the number of classes containing the corresponding item $a_i$.
\end{proof} 
\end{onlyfullversion}

\subsection{Proof of Theorem~\ref{thm:approx-delta-lam}}
In this section, we consider the case when, for each platform $p$ and item $a$, the classes containing $a$ can be partitioned into at most $\Delta$ laminar families. 
We first present a $\frac{1}{\Delta}$-approximation algorithm for the case when there is only one platform. This algorithm also generalizes the maximum independent set in hypergraphs (Proposition~\ref{prop:gmis}). We extend this algorithm to a $\frac{1}{\Delta+1}$-approximation algorithm
for the case with multiple platforms and even 
to the many-to-many setting.

\paragraph{Single platform case}

Let $G = (A \cup \{p\}, E)$ be an instance of the \CMM\ problem with a single platform $p$ and a family of classes $\mathcal{C}$ with the above restriction. 

{\em Reduction to the \gmis~problem:} We construct from $G$ an instance
of the \gmis~problem $H = (V, E_H)$ by setting $V=\{v_i\mid a_i\in A\}$ and $E_H=\{e_C\mid C\in \mathcal{C}\}$, and $f(e_C)=q(C)$.
We call a set $S\subseteq V$ {\em feasible} if for every $e\in E$ , $|S\cap e| \leq f(e)$.
We call a set $S\subseteq V$ {\em maximal} if $S$ is feasible and $S\cup \{v\}$ is not feasible for every $v\in V\setminus S$. 
Our algorithm is a simple greedy algorithm: output a maximal set of vertices. 
To prove the approximation,
we use the following lemma. 
\begin{onlyshortversion}
	See the full version of this paper for the proof.
\end{onlyshortversion}
\begin{lemma}
\label{lem:hyper-indep}
Consider a set $S\subseteq V$ and a set $B\subseteq V\setminus S$ such that $S\cup B$ is a maximal set of vertices. Then for every feasible set $C$ such that $S \subseteq C$ and $C\cap B = \emptyset$, we have $|C| \leq |S| + \Delta |B|$. 
\end{lemma}

\begin{onlyfullversion}
\begin{proof}
	We prove this by induction on $|B|$. For any pair of sets $S,B$ that satisfy the conditions of the lemma, the base case of $|B|=0$ is trivially true by maximality of $S \cup B$. 
	To prove the induction step, assume that the lemma is true for all $S,B$ where $S \cup B$ is maximal
	and  $|B| \leq k$. We will show that the lemma holds for all $S,B$ that satisfies the lemma conditions and where $|B|=k+1$. 
	Suppose, for contradiction, that we have a feasible set $C$ such that $S \subseteq C, C \cap B=\emptyset$ and
	$|C| > |S| + \Delta |B|$.
	
	Now consider the set $C \cup \{v\}$, for some $v \in B$. If $C \cup \{v\}$ is feasible,
	then we let $C' = C \cup \{v\}$, else  we construct a feasible set $C'$ as below.
	
	Consider the set of edges $E_v$ that contain $v$. We partition the set $E_v$ 
	into  $\Delta$ laminar families and call these sets $E_1,E_2,\ldots, E_\Delta$.  Since each $E_i$ is laminar and
	every edge in $E_i$ contains the vertex $v$, 
	we can arrange the edges of $E_i$ as $e_i^{(1)},e_i^{(2)},\ldots$ such that $e_i^{(j)} \subset e_i^{(j+1)}$ for all $j$. 
	Since $C$ is feasible and $C \cup \{v\}$ is not feasible, for every edge $e$ we have $|(C \cup \{v\}) \cap e| \le f(e) +1$, that is, the violation is by at most one. 
	For each $i$, we find the smallest $j$ (if any) such that $e_i^{(j)}$ has a violation of 1. We remove
	a vertex $u_i \in e_i^{(j)}$ such that $u_i  \in C \setminus B$. Note that the set $C \cup \{v\} \setminus u_i$ is feasible for 
	all the edges in $E_i$ since $e_i^{(j)} \subset e_i^{(j')}$ where $j' > j$.  Further, note that $u_i$ exists because we assumed that $S \cup B$
	and hence $S \cup \{v\}$ is feasible.  We repeat this process for each laminar family and hence may have removed at most $\Delta$
	vertices from $C \cup \{v\}$ obtaining a $C'$ which is feasible. Thus, 
	\begin{align*}
		|C'| \geq |C| + 1 - \Delta > |S| + 1 +  \Delta (|B|-1)
		.\end{align*}
	The second inequality follows from the assumption that $|C| > |S| + \Delta |B|$.
	Now consider the induction hypothesis for $S'= S \cup \{v\}$, $B'=B\setminus \{v\}$. Since $|B| = k$, and $C'$ is a feasible set containing $S'$ and is disjoint from $B'$ we have
	\begin{align*}
		|C'| \leq |S|+ 1 + \Delta (|B|-1)
	\end{align*}
	from the induction hypothesis. This is a contradiction. Thus, induction hypothesis is true and the lemma follows. 
	
\end{proof}
\end{onlyfullversion}

Let $ALG$ denote any maximal independent set of $H$ and $OPT$ be the optimal independent set. 
In the above lemma, set $S = ALG\cap OPT, B = ALG\setminus OPT, C = OPT$. The lemma implies 
$|OPT| \leq \Delta |ALG|$. 
This proves Proposition~\ref{prop:gmis}. We note that this also gives us a $\frac{1}{\Delta}$-approximation
for the \CMM~ problem in the single platform case when every item belongs to at most $\Delta$ laminar families of the platform. It is also easy to see that this algorithm runs in time $O(|V||E_H|)$. For every vertex, we add it if it does not exceed the quota of any edge it belongs to. 

\paragraph{Multiple platforms}
We can use the previous result to obtain a $\frac{1}{\Delta+1}$ approximation for the multiple platforms case via a simple $O(|E|)$-time reduction to the single platform case: For every edge $(a,p)$, make a new item $e_{a,p}$. Replace all the platforms by a single dummy platform. Since classes 
are subsets of the neighbourhood sets of items or platforms, 
they can also be seen as subsets of edges of the graph. These naturally form classes over the items $e_{a,p}$.
This combined with the result for the single platform case gives an $O(|E|\cdot |\mathcal{C}|)$ algorithm for the multiple platform case where $|\mathcal{C}|$ is the total number of classes, establishing Theorem~\ref{thm:approx-delta-lam}.

\begin{remark}[Weights on items]
	We remark that the same analysis goes through even if items have weights and the goal is to compute a maximum weight
	matching. The algorithm simply keeps matching the highest weight item that can feasibly match to the platform. 
\end{remark}

\subsection{Constant number of classes}
We can also prove Theorem~\ref{thm:approx-delta-lam} via the following general statement combined with Proposition~\ref{prop:gmis}. 
\begin{onlyshortversion}
	We leave the proof to the full version of this paper. 
\end{onlyshortversion}
We will also need this to prove Theorem~\ref{thm:constant-classes}.
\begin{lemma}\label{lemma:single-to-mul-platform}
	Given a polynomial time $\alpha$-approximation algorithm for the many-to-one matching problem with a single platform, we can obtain a polynomial time $\frac{\alpha}{1+\alpha}$-approximation for the matching problem with multiple platforms.
\end{lemma}
\begin{onlyfullversion}
\begin{proof}
	Suppose we have a set $V$ of items and a set $P$ of platforms. There is a hypergraph $G_i$ and function $f_i:V_i\xrightarrow{}\mathbb{Z}^+$ for each platform for the associated instance of \gmis. Let $OPT(G_i)$ be the set of items chosen in graph $G_i$ in OPT. Clearly, the set $OPT(G_i)$ is a feasible set in $G_i$. Note that the goal is to maximize the global number of items selected. Thus, the optimum does not necessarily pick the optimal sets in each hypergraph as there may be vertices which lie in multiple hypergraphs.
	
	Let our algorithm apply the $\alpha$-approximation algorithm to $G_1$. Let the selected independent set be $ALG(G_1)\subseteq V_1$. Now apply the $\alpha$-approximation to the graph induced on $G_2$ by the vertex set $V_2\setminus ALG(G_1)$. Then on the graph induced on $G_3$ by the vertex set $V_3\setminus ALG(G_1)\cup ALG(G_2)$ and so on. 	For $i=1,2,3,\ldots |P|,j=1,2,3,\ldots |P|$ define the sets,
	\begin{align*}
		V_i^{(0)} &= \{ v: v\in OPT(G_i), v\notin ALG(G_j)\ \forall\ j<i \} \\
		\forall j<i, V_i^{(j)} &= \{ v: v\in OPT(G_i), v\in ALG(G_j) \} 
	\end{align*}
	Then, we have
	\begin{equation}\label{optbound_alpha_approx}
		|OPT(G_i)|  = \sum_{k=0}^{i-1}|V_i^{(k)}|
	\end{equation}
	\begin{equation}\label{algbound_alpha_approx}
		|ALG(G_j)|  \geq \sum_{k=j+1}^{|P|} |V_k^{(j)}|
	\end{equation}
	because the $|V_i^{(j)}|$s form a partition over $OPT(G_i)$s and $ALG(G_j)$s. The latter may not completely be covered by the partition, and thus we have a lower bound. Note that $V_i^{(0)}$ is a feasible set in $G_i$ because it is a subset of the $OPT(G_i)$ which is a feasible set. Now consider the graph induced on $G_i$ by the vertex set $V_i\setminus \bigcup_{j<i} ALG(G_j)$. Every vertex in $V_i^{(0)}$ is present in this graph by the definition of $V_i^{(0)}$ and since $V_i^{(0)}$ is a feasible set in $G_i$, it is feasible in any subgraph of $G_i$. Thus, using our $\alpha$-approximation algorithm, we have
	\begin{equation}\label{approxbound_alpha_approx}
		|ALG(G_i)| \geq \alpha |V_i^{(0)}|
	\end{equation}
	Adding equation \ref{approxbound_alpha_approx} over all $i$ and adding equation \ref{algbound_alpha_approx} multiplied by $\alpha$ over all $j$ we have
	\begin{align*}
		(1+\alpha) \sum_{j=1}^{|P|} |ALG(G_j)| &\geq \alpha  \sum_{j=1}^{|P|}\sum_{k=j+1}^{|P|} |V_k^{(j)}| + \alpha \sum_{i=1}^{|P|} |V_i^{(0)}|\\
		(1+\alpha) \sum_{j=1}^{|P|} |ALG(G_j)| &\geq \alpha  \sum_{k=2}^{|P|}\sum_{j=1}^{k-1} |V_k^{(j)}| + \alpha \sum_{i=1}^{|P|} |V_i^{(0)}|\\
		(1+\alpha) \sum_{j=1}^{|P|} |ALG(G_j)| &\geq \alpha  \sum_{k=1}^{|P|}\sum_{j=0}^{k-1} |V_k^{(j)}|
	\end{align*}
	Using equation \ref{optbound_alpha_approx},
	\begin{align*}
		(1+\alpha) \sum_{j=1}^{|P|} |ALG(G_j)| &\geq \alpha  \sum_{k=1}^{|P|}|OPT(G_j)| \\
		\implies (1+\alpha)\ ALG &\geq  \alpha\ OPT\\
		\implies   \frac{ALG}{OPT}&\geq \frac{\alpha}{1+\alpha}
	\end{align*}
	because $ALG$ and $OPT$ are partitioned perfectly into the item set for each platform, as no item can be chosen by multiple platforms.
\end{proof}

\end{onlyfullversion}

\begin{theorem}[Formal version of Theorem~\ref{thm:constant-classes}] \label{thm:constant_classes_formal}
	The \CMM\ problem can be represented as an IP with $ 2^{\Delta}$ variables if there is only one platform with $\Delta$ classes. This can be solved in time $O(2^{2^\Delta}{\sf poly}(n))$, and also gives rise to a $\frac{1}{2}$-approximation in time $O(2^{2^\Delta}{\sf poly}(n))$ for multiple platforms, each with $\Delta$ classes of items. 
\end{theorem}
\begin{onlyshortversion}
	\begin{proof}[Sketch]
For an instance with $\Delta$ classes, every item can be represented by a $\Delta$-bit incidence vector of the classes it belongs to. This partitions the items into $2^\Delta$ {\em types}. Our ILP has one variable for each type. The runtime is obtained via \cite{lenstra_ilp}. Theorem~\ref{thm:constant-classes} follows by setting $\Delta=O(1)$. In practice, there exist integer program heuristics that are faster \cite{fischetti2005_IP,balas2001_IP}. For arbitrarily large number of platforms, with a constant number of classes in each platform, we can use this with Lemma \ref{lemma:single-to-mul-platform} to get a $\frac{1}{2}$-approximation. 
See the full version of this paper for the details.
\end{proof}
\end{onlyshortversion}
\begin{onlyfullversion}
\begin{proof}
We first reduce our instance of \CMM\ with $1$ platform to an instance of \gmis. Let the hypergraph corresponding to the platform be $H$. Let the set of hyperedges of $H$ be $E$. Define $2^{E}=\{ S~|~ S \subseteq E \}$ to be the power set of the edge set. Let the quota of a hyperedge $e\in E$ be $q_e$. 
For $S\in 2^E$, let $x_{S}$ be the number of vertices we pick from the set
\[
\bigcap_{e\in S} e \cap \bigcap_{e\in E\setminus S}  \overline{e}
\]
where $\overline{e} = V(H)\setminus e$. Then we have our ILP as
\[
\text{Maximize }\sum_{S\in 2^{E}} x_S
\]
subject to 
\begin{align*}
	\forall~ e\in E,&\\
	& \sum_{S\in 2^{E}: e\in S} x_{S} \leq q_{e}\\
	\forall~ S\in 2^{E}, &\\
	& x_{S} \leq \min\left(q_S, \left|\bigcap_{e\in S} e \cap \bigcap_{e\in E\setminus S}  \overline{e}\right|\right)\\
	\forall~ S\in 2^E,&~  x_S \in \{0,1\}
\end{align*}	

The above ILP has $2^{|E|}$ variables, and $|E|+ 2^{|E|}$ constraints. It is known that an ILP with $n$ variables and $m$ constraints can be solved in time exponential in $n$ and polynomial in $m$ \cite{lenstra_ilp}. Thus, when $|E|= O(\log \log n)$, where $n=|A|$, the number of items, we can find the optimum solution in time polynomial in the input size. If we have that the number of posts $|P|$ is greater than 1 but still a constant, then we can keep one variable $x_{S,p}$ for every type-post pair. We can similarly proceed and get a polynomial-time algorithm. 
\end{proof}
\end{onlyfullversion}

\subsection{Bounded Average degree}
We extend the result from the previous section for a single platform to the case when the average number of laminar families of classes an item belongs to is bounded by $\Delta$. We state it in terms of \gmis  here.
Now consider the case where the hypergraph $H$ constructed above has only bounded average degree of $\Delta$.

\begin{proof}[Proof of Theorem~\ref{thm:approx-avg-degree}]
Since the average degree is $\Delta$, for any $f$, there cannot be more than $\frac{n}{f}$ vertices of degree more than $f\Delta$.  Suppose we estimate $r$ and set $f = \frac{2}{r}$.
We call a vertex {\em low degree} if its degree is at most $f \Delta$, otherwise the vertex is {\em high degree}. Then the number of low degree vertices is $ \geq n \left(1 -\frac{r}{2}\right)$. In the graph induced by the low degree vertices, the size of the optimal independent set is at least $\frac{OPT}{2}$, since at 
most  $\frac{OPT}{2}$ vertices of high degree.
We use our $\frac{1}{\Delta}$ approximation algorithm on the graph induced by the low degree vertices. 
Since this graph has maximum degree $\leq \frac{2\Delta}{r}$, the size of the independent set has size $\geq \frac{r\cdot OPT}{4\Delta}$. 

Thus, our approximation ratio is at least $\frac{r}{4\Delta}$. We finally need to estimate $r$. We guess 
a value of $OPT$ from 1 to $n$ and run the above procedure for each of the guesses. 
Amongst all the solutions that we obtain, we pick the one with the highest cardinality. 
This is guaranteed to do at least as well as the case when we picked the correct value of $OPT$. 
\end{proof}

%% file: 4online.tex
\section{Online algorithms}\label{sec:online}

The online algorithm for Theorem~\ref{thm:approx-online-delta-lam} is essentially the same as the one in Proposition~\ref{prop:gmis} and works even for an arbitrary input model.
Whenever an item arrives online, we match it to a platform such that the matching remains feasible. If there is no such platform, we leave it unmatched. The output is a maximal set, which by Proposition~\ref{prop:gmis}, gives us the required competitive ratio. 
However, we point out that this only works for the unweighted version. 

Since the \CMM\ problem is \NP-Hard in general, we also give online algorithms for the case where the classes form a laminar family. This version is known to have an efficient offline algorithm, via the construction of a simple flow network. A similar construction is used to compute the rank-maximal matching in \cite{nasre_classified_matchings}. 

In this setting, we study the many-to-many \CMM\ model (Model \ref{mod:mm}), when the classes are laminar. 
We assume an input model where the item set arrives in a uniformly random permutation (Model \ref{mod:random}). 
For the sake of the analysis, we assume that a random variable $y_i$ picked uniformly
at random from $[0,1]$ for every item $a_i$, and the items arrive in the increasing order of $y_i$. 
Therefore the random vector $\Vec{y} :=(y_1,y_2,\ldots,y_n)$ fully describes the order of arrival of the items. We use $\Vec{y}_{-i}$ to represent the vector after removing $y_i$ from $\Vec{y}$. 
We use the following greedy algorithm, and analyze its competitive ratio (in expectation): {Keep an arbitrary, fixed ranking of all the platforms in $P$. When an item arrives online, match it to as many platforms as possible, picking the highest ranked ones.}
%
%
\begin{onlyshortversion}
	\begin{proof}[Sketch of Theorem \ref{thm:approx-online-random}]
\end{onlyshortversion}

We use a linear programming relaxation of our problem to analyze our algorithm. We set the primal values according to the output of our algorithm, thereby ensuring the feasibility of the primal solution. Now we need to construct an appropriate dual solution. We use the following folklore fact about the well-known method of dual fitting in designing algorithms. 
This technique is used in \cite{devanur_ranking,huang_deadlines} among others.

In the primal LP, we have a variable $x_{ij}=1 \iff$ item $a_i$ is matched to platform $p_j$. We also have constraints for both the item and platform classes. In the dual LP, we have variables corresponding to constraints in the primal LP. 
\begin{onlyshortversion}
We describe the LP formally in the full version of this paper. 
\end{onlyshortversion}
\begin{onlyfullversion}
	Let the dual variables corresponding to the item and platform classes $C_i^{(k)}$ and $C_j^{(k)}$ be $\alpha_i^{(k)}, \beta_j^{(k)}$ respectively. 
	\paragraph{Primal}
\begin{equation*}
	\text{Maximize} \sum_{j:p_j\in P}\sum_{i: (a_i,p_j)\in E}x_{ij}
\end{equation*}
Subject to 
\begin{align*}
	\forall~ i:a_i\in A, \forall~ k: C_i^{(k)} &\sum_{j:(a_i,p_j) \in E, p_j\in \Ca_i^{(k)}} x_{ij} \leq q(\Ca_i^{(k)})\\
	\forall~ j: p_j\in P, \forall~ k: C_j^{(k)} &\sum_{i:(a_i,p_j) \in E, a_i\in \Cp_j^{(k)}} x_{ij} \leq q(\Cp_j^{(k)})\\
	\forall~ i:a_i\in A,j:p_j\in P, &\qquad 0\leq x_{ij} \leq 1
\end{align*}
\paragraph{Dual}
\begin{equation*}
	\text{Minimize } \sum_{i:a_i\in A}\sum_{k} \alpha_i^{(k)} q(\Ca_i^{(k)})+  \sum_{j:p_j\in P}\sum_{k} \beta_j^{(k)} q(\Cp_j^{(k)})
\end{equation*}
Subject to 
\begin{align*}
	\forall~ (a_i,p_j) \in E,
	&\sum_{k:p_j\in \Ca_i^{(k)}} \alpha_i^{(k)}+ \sum_{k:a_i\in \Cp_j^{(k)}} \beta_j^{(k)} \geq 1\\
	\forall\ a_i\in A, \qquad \alpha_i^{(k)} \geq 0 &\quad\qquad \forall\ p_j\in P,\forall\ k, \qquad  \beta_{j}^{(k)} \geq 0
\end{align*}
\end{onlyfullversion}

\begin{fact}[Folklore]\label{fact:dual_fitting}
	Suppose we can set the dual variables such that the primal objective is equal to the dual objective, 
	and the dual constraints of the form $\alpha \geq 1$ instead satisfy $\E[\alpha]\geq F$. Then, our algorithm has a competitive ratio $F$ in expectation.
\end{fact}

\begin{onlyfullversion}
\begin{proof}
	Consider the dual solution obtained by setting $\hat{\alpha}_i^{(k)} = \frac{\E[\alpha_i^{(k)}]}{F}, \hat{\beta}_j^{(k)} = \frac{\E[\beta_j^{(k)}]}{F}$. Clearly this is a feasible solution. Further, we have from weak duality
	\begin{align*}
		\sum_{a_i\in A}\sum_{k} \hat{\alpha}_i^{(k)} q(\Ca_i^{(k)}) + \sum_{p_j\in P}\sum_{k} \hat{\beta}_j^{(k)} q(C_j^{(k)}) &\geq OPT.
		\intertext{Taking the expectation over the randomness of the input on both sides}
		\E\left[\sum_{a_i\in A} \sum_k \frac{\alpha_i^{(k)} q(\Ca_i^{(k)})}{F} + \sum_{p_j\in P}\sum_{k} \frac{\beta_j^{(k)}q(\Cp_j^{(k)})}{F}\right] &\geq OPT\\
		\implies \E\left[\sum_{j\in P}\sum_{i: (i,j)\in E}x_{ij}\right]   \geq  F \cdot OPT
		.\end{align*}
	Thus, the dual solution will certify a competitive ratio of $F$ in expectation. 
\end{proof}
\end{onlyfullversion}
Now, we need to set the dual variables so that the dual constraints have a lower bound of $F$. 
\begin{onlyshortversion}
	The details are left to the full version. 
\end{onlyshortversion}
\begin{onlyfullversion}
	Let $g:[0,1]\to [0,1]$ be a monotonically increasing function to be fixed later. 
	Initially, all dual variables are zero. When a new item arrives, we update the primal
	solution based on whether we matched the item or now, and update the dual solution as follows.
	
	\begin{itemize}
		\item If an edge $(i,j)$ is chosen by the algorithm (i.e. $x_{ij} = 1$) then set $\alpha_i^{(j)} := 1$, where $\alpha_i^{(j)}$ is the dual variable corresponding to the class $C_i^{(j)}=\{p_j\}$. 
		\item If any class $\Cp^{(k)}_j$ of platform $j$ is tight then for all items $a_i$ in the class $\Cp_j^{(k)}$
		set $\alpha_i^{(j)} =g(y_i)$ if we had that $\alpha_i^{(j)}=1$ before the new item arrived. Fix $\beta_j^{(k')}=0$ for all $k'$ such that $\Cp_j^{(k')}\subset \Cp_j^{(k)}$. Let it remain unchanged henceforth. Also set 
		\[ \beta_j^{(k)} := 1 - \frac{ \sum_{i: i \in \Cp_j^{(k)}, \alpha_i\neq 0} g(y_i)    } {q(\Cp_j^{(k)})}   \geq 1 - \max_{i \in \Cp_j^{(k)}, \alpha_i\neq 0} g(y_i) .\]
		\item If any class $\Ca_i^{(k)}$ of item $a_i$ is tight then  set
		\begin{equation*}
			\alpha_i^{(k)} := \frac{1} {q(\Ca_j^{(k)})} \sum_{k':\Ca_i^{(k')}\subset \Ca_i^{(k)}}\alpha_i^{(k')}q(\Ca_i^{(k')}) \geq  g(y_i)
			.\end{equation*}
		Fix $\alpha_i^{k'}=0$ for all $k'$ such that $\Ca_i^{(k')}\subset \Ca_i^{(k)}$. Let it remain unchanged henceforth. 
	\end{itemize}
	
	It can be easily seen that the dual objective function is equal to the primal 
	objective function. It remains to be shown that the dual constraints are always greater 
	than some $F$ in expectation.
	
\end{onlyfullversion}
The following is the key lemma in analyzing how the algorithm behaves depending on $\Vec{y}$.
Although it is inspired by \cite{devanur_ranking,huang_deadlines,mehta_online_random}, 
our many-to-many model (Model \ref{mod:mm}) is more complicated in that moving one vertex up the ranking can cause more changes to the matching because an item can match to multiple platforms. Even apart from the platform classes, we must take care of item classes as well. 
\begin{onlyshortversion}
To that end, we show the following lemma whose proof we leave to the full version of this paper. 
\end{onlyshortversion}
\begin{onlyfullversion}
To that end, we show the following lemma. 
\end{onlyfullversion}
\begin{lemma}
	\label{lemma-item}
	For any $i,j$ such that $j\neq i$, if an item $a_j $ is matched to some platforms at $y_i=1$, then it cannot be unmatched from any platform at $y_i=\theta \in [0,1]$ due to an item class.
\end{lemma} 
\begin{onlyfullversion}
	\begin{proof} 
	Suppose the platforms are ranked $p_1\succ p_2 \succ p_3 \ldots$ and the items arrive in the order $a_1,a_2,\ldots$ so that $y_1<y_2\ldots$. We first fix an $i$. 
	We will prove this on induction for $j$. That is, we first show that the lemma is true for item $a_1$ (ie $j=1$) and then inductively argue that it has to be true for every $j$. 
	
	For the base case, suppose $a_1$ is matched to $p_u$ when $y_i=1$.
	Then if $y_i\in (y_1,1]$, item $a_1$ is unaffected and the statement is true. Now let $y_i<y_1$.
	We consider the changes that $y_i$ brings. 
	The only way $a_1$ can get unmatched from some platform due to an item class is if the following chain of events ocurred.
	$a_i$ matches to some platform $p_u$, forcing $a_1$ to unmatch from $p_u$ due to a platform class.
	$a_1$ matches to a platform $p_v$, which it couldn't do before because $p_u,p_v$ belonged to a tight item class $C_1$.
	$a_1$ is forced to unmatch from some platform $p_w$ because $p_v,p_w$ belong to a tight item class $C_2$. 
	Note that both $C_1,C_2$ contain $p_v$ and by laminarity, one class must contain the other. We can argue that in either case, we get a contradiction. 		
	
	Suppose the induction hypothesis is true up to (but not including) some $j$. Then we will show that it is true for $a_j$ as well. Suppose not. Then for $y_i=\theta$, $a_j$ is unmatched from some $p_u$ due to an item class $C_1$. Then, there is a platform $p_v\in C_1$ such that $a_j$ was not matched to $p_v$ at $y_i=1$, but replaced $p_u$ at $y_i=\theta$. This replacement can happen only if $p_v \succ p_u$. 
	
	Why was $p_v$ not matched at $y_i=1$? Through a similar argument as in the base case, it can be seen that it is not possible due to an item class of $a_j$. Thus, $p_v$ could not match to $a_j$ at $y_i=1$ due to a platform class (let it be $C_2$). Since it could match at $y_i=\theta$, there must be some item $a_k$ that was matched to $p_v$ at $y_i=1$ but not at $y_i=\theta$. Then, it must be that $y_k<y_j$. 
	
	Why was $a_k$ unmatched from $p_v$ at $y_i=\theta$? By the induction hypothesis, it must be due to a platform class of $p_v$. Let it be $C_3$. Then there is an item $a_l$ that is matched to $p_v$ at $y_i=\theta$ but not at $y_i=1$, such that $y_l<y_k$. Now consider $C_2,C_3$. Since $a_k\in C_2\cap C_3$, one class must contain the other by laminarity. We can argue that in both cases, we get a contradiction.
\end{proof}

\end{onlyfullversion}

Once we have the lemma, we show Theorem~\ref{thm:approx-online-random} the following way. Let $\alpha \geq 1$ be such a dual constraint. We want to show that $\E\left[ \alpha \right ] \geq F$ or equivalently, $\E_{y_{-i}}\left[ \E_{y_i} \left[ \alpha \right] \right ] \geq F$. We look at the inner expectation. We fix $y_{-i}$ and vary $y_i$ from 0 to 1, and show using dual-fitting arguments and Lemma~\ref{lemma-item} that $\E\left[ \alpha \right ] \geq \left (1-\frac{1}{e}\right )$. From Fact~\ref{fact:dual_fitting}, this completes the proof of  Theorem~\ref{thm:approx-online-random}.
\begin{onlyfullversion}
	
\begin{proof}[Proof of Theorem~\ref{thm:approx-online-random}]
Consider a fixed platform $p_j$. Suppose that some item class that $a_i$ is in is tight until $y_i=\theta$. Thus, we have
\begin{equation*}
	\sum_{k: p_j\in C_i^{(k)}} \alpha_1^{(k)} \geq g(y_1)
	.\end{equation*}
until $y_i=\theta$. Afterwards, we have two cases, 
\begin{enumerate}
	\item Case 1: At $y_i=\theta$, $a_i$ matches to $p_j$.
	Note that $\sum_k \alpha_i^{(k)} \geq g(y_i)$ remains true as long as $a_i$ is matched to $p_j$.
	\begin{enumerate}
		\item Case 1(\textit{a}):  $a_i$ remains matched to $p_j$ until $y_i=1$. Then the previous inequality $\sum_k \alpha_i^{(k)} \geq g(y_i)$ continues to hold and we have
		\begin{align*}
			\E_{y_i}\left[\sum_k\alpha_i^{(k)}\right] \geq  \int_{0}^{1} g(y_i) d y_i
			.\end{align*}
		\item Case 1(\textit{b}): $a_i$ becomes unmatched from $p_j$ at some point when $y_i=\theta'>\theta$ due to a class of $p_j$. In this case some platform class (say $\Cp^{(1)}$) must have been tight. We claim that in this case, at least one class of $p_j$ that $a_i$ belongs to is tight regardless of the value of $y_i$. 
		
		Clearly, increasing the value of $y_i$ beyond $\theta'$ does not affect the tightness of $\Cp^{(1)}$. Thus, the class is tight even at $y_i=1$. Now at any other value of $y_i$, if the class is not tight then some item must've been replaced. By lemma \ref{lemma-item}, it must have been due to a platform constraint (say $\Cp^{(2)}$). In this case, we have either 
		\begin{itemize}
			\item $\Cp^{(1)} \subset \Cp^{(2)}$, in which case $a_i$ also belongs to this tight class. 
			\item $\Cp^{(2)} \subset \Cp^{(1)}$, in which case the replacing element does not reduce the number of elements matched in $\Cp^{(1)}$ (and cannot relax the constraint). 
		\end{itemize}
		Since the highest rank of the items matched to $p_j$ in the tight class (either $\Cp^{(1)}$ or $\Cp^{(2)}$) is at most $\theta'$, and $\sum_k \alpha_i^{(k)} \geq g(y_i)$ remained true until $y_i=\theta'$ we have
		\begin{equation*}
			\E_{y_i}\left[\sum_k\alpha_i^{(k)}+\sum_k\beta_j^{(k)}\right] \geq  \int_{0}^{\theta'} g(y_i) d y_i + 1-g(\theta')
			.\end{equation*}
		\item Case 1(\textit{c}):  $a_i$ becomes unmatched from $p_j$ at some point when $y_i=\theta'>\theta$ due to a class of $a_i$. We will show that this case is not possible. 
		
		Suppose that at $y_i=\theta'>\theta$, item $a_i$ got matched to $p_{k}$, which prevented $a_i$ from matching to $p_j$. Thus, $p_j,p_{k} \in \Ca^{(1)}$ for some class which was tight. Since $a_i$ could not match to $p_j$ at $y_i=\theta$, it must be that it was prevented from doing so owing to a item class. If it was due to a platform class, then it would not be possible for it to match at $y_i=\theta'$ since it now has a worse rank than before. 
		
		Thus, there is some $p_{l}$ such that $a_i$ was matched to $p_l$ at $y_i=\theta$ (but not at $y_i=\theta'$) and $p_l,p_{k} \in \Ca^{(2)}$ for another class which was tight. Again, due to laminarity and the non-empty intersection of $\Ca^{(1)}$ and $\Ca^{(2)}$ we have
		\begin{itemize}
			\item If $ \Ca^{(1)} \subset \Ca^{(2)} $ then it should not have been possible to match $p_j$ at $y_i=\theta$ because $\Ca^{(2)}$ was tight.
			\item If $ \Ca^{(2)} \subset \Ca^{(1)}$, then $p_k$ is not the replacing element since matching $p_k$ and unmatching $p_l$ cannot relax the constraint.
		\end{itemize}
	\end{enumerate}
	\item Case 2: $a_i$ never matches to $p_j$. We assumed that some item class was tight till $y_i=\theta$. After that, since it was not possible for $a_i$ to match to $p_j$, then some class of $p_j$ that $a_i$ belonged to was tight. By a similar logic as before, it can be seen that a class is tight regardless of the value of $y_i$ and the highest rank of any item in the class is $\theta$. Thus, we have
	\begin{align*}
		\E_{y_i}\left[\sum_k\alpha_i^{(k)}+\sum_k\beta_j^{(k)}\right] \geq  \int_{0}^{\theta} g(y_i) d y_i + 1-g(\theta)
		.\end{align*}
	Note that the case where $\theta=1$ is when some class of $a_i$ always prevented $a_i^j$ from matching to $p_j$. In this case, like before we have $\E\left[\sum_k\alpha_i^{(k)}\right] \geq  \int_{0}^{1} g(y_i) d y_i$.
\end{enumerate}
Thus, we have that
\begin{align*}
	&\E_{\Vec{y}_{-i}}\left[\E_{y_i}\left[\sum_k\alpha_i^{(k)}+\sum_k\beta_j^{(k)}\right]\right] \\
	&\geq \min\left(\int_{0}^{1} g(y_i) d y_i, \min_{\theta\in[0,1]}\left(\int_{0}^{\theta}  g(y_i) d y_i + 1-g(\theta)\right)\right)
	\end{align*}	
as required.
\end{onlyfullversion}
\end{proof}

%% file: 7experiments.tex
\section{Experiments}\label{sec:experiments}
In this section, we present the experimental evaluation of our offline algorithms from Theorem~\ref{thm:constant-classes} and Theorem~\ref{thm:approx-delta-lam}. We use a total of seven datasets which we categorize as real-world and synthetic datasets.
The three real-world datasets are sourced from an elective allocation process at an educational institution.
The four synthetic datasets are generated as described below. All experiments were run on a laptop running on a 64-bit Windows 10 Home edition, and equipped with an Intel Core i7-7500U CPU @2.7GHz and 12GB of RAM. For solving integer programs, we used IBM ILOG CPLEX Optimization Studio 20.1 through its Python API. All code was written to run on Python 3.8. 

\begin{table}[tb]
	\centering
	\begin{tabular}{@{}|l|r|r|r|@{}}
		\toprule
		Dataset &$ \frac{1}{2} $-approx 	& $\Delta$-approx & OPT \\ \midrule
		Real-1     & 1871.5 \hfill (0.92)         & 1899.8 \hfill (0.93)  & 2035     (1)   \\ 
		Real-2      & 1988.6 \hfill (0.92)        & 2014.0  \hfill (0.93)    & 2170    (1)     \\ 
		Real-3       & 1938.6 \hfill (0.92)       & 1936.7 \hfill (0.92)     & 2107       (1)   \\ \bottomrule
	\end{tabular}
\caption{Comparison of (average) solution values on the real-world datasets. Relative values are in parentheses. }
\label{tab:1}
\end{table}

\begin{table}[t]
	\centering
		
	\begin{tabular}{@{}|l|r|r|r|@{}}
		\toprule
		Dataset      &	 \multicolumn{1}{c|}{$\frac{1}{2} $-approx}  &   \multicolumn{1}{c|}{  $\Delta$-approx} 		& 		 \multicolumn{1}{c|}{OPT} \\ \midrule
Real-1       & 0.39   \hfill (1.23)      & 0.11     \hfill (4.29)   & 0.48      (1)    \\ 
Real-2       & 0.43   \hfill (1.03)      & 0.11    \hfill (3.89)    & 0.44     (1)      \\ 
Real-3       & 0.33    \hfill (1.23)     & 0.10    \hfill (3.90)   & 0.40      (1)    \\ \bottomrule
	\end{tabular}
	\caption{Comparison of (average) running-times in seconds on the real-world datasets. Relative speedups are in parentheses.}
\label{tab:2}
\end{table}

\begin{table}[t]
	\centering 
	\begin{tabular}{@{}|l|r|r|r|@{}}
		\toprule
		Dataset      &	 \multicolumn{1}{c|}{$\frac{1}{2} $-approx}  &   \multicolumn{1}{c|}{  $\Delta$-approx} 		& 		 \multicolumn{1}{c|}{OPT} \\ \midrule
		large-dense    & 239552  \hfill (0.97)  & 239566.4 \hfill (0.97)   & 247537  (1)     \\ 
		large-sparse   & 212600.1 \hfill (0.97)  & 211885.1 \hfill (0.97)  & 218622  (1)     \\ 
		small-sparse    & 72676.4  \hfill (0.93) & 72821.5  \hfill (0.93)   & 78279   (1)     \\ 
		small-dense     & 75887.7  \hfill (0.95) & 76133.4  \hfill (0.95)   & 79827  (1)      \\ \bottomrule
	\end{tabular}
	\caption{Comparison of (average) solution values in the synthetic datasets. Relative values are in parentheses.}
	\label{tab:3}
\end{table}

\begin{table}[t]
	\centering 	
	\begin{tabular}{@{}|l|r|r|r|@{}}
		\toprule
		Dataset      &	 \multicolumn{1}{c|}{$\frac{1}{2} $-approx}  &   \multicolumn{1}{c|}{  $\Delta$-approx} 		& 		 \multicolumn{1}{c|}{OPT} \\ \midrule
		large-dense     & 5.68 \hfill (14.41)& 2.90  \hfill (28.21)        & 81.99  (1)    \\ 
		large-sparse    & 4.67 \hfill (15.14)& 2.19  \hfill (32.19)        & 70.73  (1)     \\ 
		small-sparse    & 1.55  \hfill (3.00)& 0.46   \hfill (10.07)      & 4.68   (1)     \\ 
		small-dense      & 1.73 \hfill (5.39)& 0.58   \hfill (16.14)        & 9.37  (1)     \\ \bottomrule
	\end{tabular}
	\caption{Comparison of (average) running-times in seconds in the synthetic datasets. Relative speedups are in parentheses.} 
	\label{tab:4}
\end{table}

\paragraph{Real-world datasets:} We use data from three course-registration periods at an educational institution. Each dataset has around 100 courses
and 2000 students. The students and the courses correspond to items and platforms respectively in our model. 
The edges represent the courses that a student is interested in. 
The students are partitioned into 13 departments (majors) as well as 5 batches (1st year--5th year). 
Each course has an overall quota denoting the maximum number of students that can be allotted to it.
For each course, we {introduce} a quota for each department and a quota
for each batch. 
Each course belongs to one of two categories, and each student can be matched to at most one course of each category. The goal is to maximize the number of edges selected subject to these constraints.
This can be immediately  viewed as an instance of CMM.

\paragraph{Synthetic Datasets:}
Modelled on the real-world datasets, we synthetically generate large instances and compare the performance of our algorithms to
the optimal algorithm implemented using a matching Integer Linear Program. The synthetic datasets are generated as follows.
Datasets labelled `large' have 500 courses, and 20 departments with 10,000 students in each department.
The  datasets labelled `small' have 300 courses, and 20 departments with 2,000 students in each department. 
The students have a degree that is chosen uniformly at random between 3 and 10 in the `dense' datasets and between 3 and 5 in the `sparse' datasets. 
Students choose their courses randomly based on a non-uniform probability distribution. This distribution is 
defined by assigning a random `popularity' value to each course. We observe this feature in the real-world dataset, 
where all courses are not equally popular. We also experiment without this feature, and obtain similar results. 
\

We compare our performance and running-time with the optimal solution obtained by solving the standard Matching ILP augmented with the constraints for each class. 
All running-times include the time taken for file I/O. 
The solution values and running-times are averaged over 10 runs. Though our algorithms are deterministic, 
these implementations utilize some randomness because of the use of hash-tables. 
Observe that since we have two laminar families of classes, Theorem~\ref{thm:constant-classes} and Theorem~\ref{thm:approx-delta-lam} provide theoretical guarantees of only $ \frac{1}{2} $  and $ \frac{1}{3} $ respectively. 
However, the performance of the algorithms on both real-world and random data are close to optimal. All our tables provide absolute values of the solution value and running-time of the algorithm from 
Theorem~\ref{thm:constant-classes} (column $\frac{1}{2}$-approx) and algorithm from Theorem~\ref{thm:approx-delta-lam} (column $\Delta$-approx), as well as the relative value or relative speedup in comparison to that of the Matching ILP (column OPT). 

\subsection{Observations}

Table~\ref{tab:1} and Table~\ref{tab:2} provide the solution values and running times for real-world instances whereas Table~\ref{tab:3} and Table~\ref{tab:4}
provide the same for the synthetically datasets.
In both the real-world and synthetic datasets, both of our algorithms output solutions with value at least 90\% of the optimum value. This seems to suggest that both real-world or random settings are `easier' than the worst-case instances for our algorithms. Furthermore, we believe that the significantly improved running-time more than makes up for loss of 10\% in the output value. The biggest speedups are observed in the `large' datasets, where our algorithms achieve speedups of $15\times$ and $30\times$ respectively. This is expected because the ILP takes time exponential in the size of the graph. 

\begin{onlyfullversion}
\subsection{Additional Synthetic Datasets}\label{app:experiments}

In addition to the previous results, we evaluate the performance of our algorithms on synthetically
generated datasets where the adjacency for each vertex was generated uniformly at random. That is, we did not use the `popularity' 
measure of the courses when selecting adjacent courses for a student. 
Like before, datasets labelled `large' have 500 courses, 
and 20 departments with 10,000 students in each department. The datasets labelled `small' have 300 courses, 
and 20 departments with 2,000 students each in each department. 
In the datasets labelled `dense', students have a degree 
that is chosen uniformly at random between 3 and 10, and in datasets labelled `sparse', students have a degree that is
chosen uniformly at random between 3 and 5. Students choose their courses randomly based on a uniform probability distribution over the courses, while ensuring that they apply to at least one course from each category. 
Table~\ref{tab:5} and Table~\ref{tab:6} provide the solution values and running times for the synthetically generated datasets described here. 
We observe a similar pattern here. Our algorithms achieve at least 95\% of the optimal solution in each case, while being faster than the Matching ILP (OPT).

\begin{table}[tb]
	\centering
	\begin{tabular}{@{}|l|r|r|r|@{}}
		\toprule
		
		Dataset      &   \multicolumn{1}{c|}{$\frac{1}{2} $-approx}  &   \multicolumn{1}{c|}{  $\Delta$-approx}                 &                \multicolumn{1}{c|}{OPT} \\ \midrule
		
		large-dense  & 245650.0  \hfill(1)    & 245650.0  \hfill(1)    & 245650  (1)     \\ 
		large-sparse & 251420.9 \hfill(0.99)   & 251304.1 \hfill(0.99)   & 251840  (1)     \\ 
		small-sparse & 74850.7  \hfill(0.95)   & 74914.4  \hfill(0.95)   & 78692    (1)    \\ 
		small-dense  & 77839.0    \hfill(0.97)   & 77848.1  \hfill(0.97)   & 79849    (1)    \\ \hline
	\end{tabular}
	\caption{Comparison of (average) solution values for the second set of synthetic datasets. }
	
	\label{tab:5}
\end{table}

\begin{table}[tb]
	\centering
	\begin{tabular}{@{}|l|r|r|r|@{}}
		\toprule
		
		Dataset      &   \multicolumn{1}{c|}{$\frac{1}{2} $-approx}  &   \multicolumn{1}{c|}{  $\Delta$-approx}                 &                \multicolumn{1}{c|}{OPT} \\ \midrule

		large-dense  &  5.91  \hfill(24.97)    &3.05 \hfill(48.41)      & 147.72 (1)    \\ 
		large-sparse &  4.74  \hfill(17.52)    &2.30  \hfill(36.09)     & 83.07   (1)   \\ 
		small-sparse &  1.60  \hfill(2.40)    &0.50  \hfill(7.62)     & 3.83   (1)    \\ 
		small-dense  &  1.71  \hfill(4.53)    &0.58  \hfill(13.29)     & 7.78    (1)   \\ \hline
	\end{tabular}
	\caption{Comparison of (average) running times for the second set of synthetic datasets. All times are in seconds.}
	\label{tab:6}
\end{table}

\end{onlyfullversion}

%% file: 5future.tex
\section{Conclusion}

In this paper we gave approximation algorithms for the \CMM\ problem in various offline and online settings. Improving these approximation factors or showing matching lower bounds are natural open questions. There are existing algorithms that break the $1-1/e$ barrier for online bipartite matching (without group fairness constraints) \cite{mehta_online_iid,karande2011_online_matching}; obtaining similar bounds for online CMM is also an interesting open problem.
